\documentclass[conference]{IEEEtran}
\usepackage{graphicx, epstopdf}
\usepackage{amssymb, amsmath}
\usepackage{commath}
\usepackage{setspace}
\usepackage{acronym}
\usepackage{mathrsfs}
\usepackage{ifthen}
\usepackage{textcomp}
\usepackage{float}
\usepackage{subfigure}
\usepackage{color}
\usepackage{cite}
\usepackage{amsthm}

\newtheorem{theorem}{Theorem}

\newtheorem{proper}{Property}

\newtheorem{defn}{Definition}

\DeclareMathOperator\erfc{erfc}
\def\sgn{\mathop{\rm sgn}\nolimits}

\IEEEoverridecommandlockouts

\begin{document}
\bstctlcite{ICC09_Ref2:BSTcontrol}

\title{Stable Distributions as Noise Models \\for Molecular Communication}

\author{ \IEEEauthorblockN{Nariman~Farsad$^1$,
											Weisi~Guo$^2$,
											Chan-Byoung~Chae$^3$,
											 and Andrew~Eckford$^4$}

\IEEEauthorblockA{$^1$ Department of Electrical Engineering, Stanford University, USA.\\
$^2$ School of Engineering, University of Warwick, Coventry, UK.\\
$^3$ School of Integrated Technology, Yonsei University, Korea.\\
$^4$ Department of Electrical Engineering and Computer Science, York University, Canada.\\
Corresponding Email: nfarsad@stanford.edu.
}
}
\maketitle

\begin{abstract}
In this work, we consider diffusion-based molecular communication timing channels. Three different timing channels are presented based on three different modulation techniques, i.e., i) modulation of the release timing of the information particles, ii) modulation on the time between two consecutive information particles of the same type, and iii) modulation on the time between two consecutive information particles of different types. We show that each channel can be represented as an additive noise channel, where the noise follows one of the subclasses of stable distributions. We provide expressions for the probability density function of the noise terms, and numerical evaluations for the probability density function and cumulative density function. We also show that the tails are longer than Gaussian distribution, as expected.
\end{abstract}


\section{Introduction}

Molecular communication is a biologically inspired form of communication, where chemical signals are used to transfer information \cite{Farsard_arXiv14}. It is possible to modulate  information on the information particles using different techniques such as: concentration \cite{kur12}, type \cite{kim13}, number \cite{far12NanoBio}, or time of release \cite{eck07}. Moreover, information particles can propagate from the transmitter to the receiver using diffusion \cite{mah10}, active transport \cite{far14TSP}, bacteria \cite{lio12}, and flow \cite{bic13}. Recently the possibility of molecular communication has been demonstrated using a tabletop experimental setup \cite{far13,far14INFOCOM}.

We consider diffusion-based molecular communication, where information is encoded on the time of release of molecules. Timing channels for diffusion-based molecular communication were first proposed in \cite{eck07}. A molecular communication timing channel based on additive inverse Gaussian distributed noise for flow induced channels was presented in \cite{sri12}, and tight bounds for the capacity of this channel was presented in \cite{li14}. In \cite{ata12}, a special type of timing modulation, where the order of release of consecutive molecules of different type is used to encode information is proposed. The time interval between release of two consecutive release of large number of information particles is proposed as a modulation scheme in \cite{kri13}. 

In this work we propose three general classes of timing channels for diffusion-based molecular communication: the regular timing channel, where information is encoded in the release timing of information particles (channel A); time between release modulation using the same type of information particles, where the information is encoded in the time between release of two consecutive particles of the same type (channel B); and time between release modulation using different types of information particles (channel C). In all three cases it is demonstrated that the channel can be reduced to an additive noise channel where the noise term falls in the stable distribution family \cite{nol15}.  In particular, for channel A the noise follows the well-known L\'evy distribution.

Stable distributions have been used in a number of fields to model noise. In \cite{he14}, alpha-stable distributed noise was used to create a more realistic noise model for room acoustic channels. In radio communications, symmetric alpha-stable distributions were used to model impulsive non-Gaussian noise that exists in some systems such as ultra-wide bandwidth (UWB) systems \cite{nir09,fan12}. Capacity bounds for a special class of alpha-stable additive noise channels had been provided in \cite{wan11,fah12}. 

There are only three classes of stable distribution with closed-form expressions for the probability density function (PDF) in term of elementary functions: Gaussian, Cauchy, and L\'evy. In this work, we derive closed-form expressions for the PDF of the noise terms in our channels in terms of the complex error function and Voigt functions \cite{abr11}, which are used in other fields of science such as physics. We numerically compare the stable-distributed noise densities and distribution functions to the Gaussian distribution, and show that the stable distribution exhibits longer tails. We present expressions for the asymptotic tail probability of the noise models and show that the expressions converge to the actual tail probabilities quickly.

The rest of this paper is organized as follows. In Section \ref{sec:model} we present three timing channel models for diffusion-based molecular communication. We then derive the PDF for the additive noise term in each channel model in Section \ref{sec:noiseModel}. Numerical evaluations of the PDF and the cumulative distribution function are presented in Section \ref{sec:numEval}, and expressions for the tail probabilities are provided. The concluding remarks are presented in Section \ref{sec:conc}.


\section{Timing Channel Models}
\label{sec:model}
In this section we present three different timing channels based on three different timing modulation schemes for diffusion based molecular communication systems. In our models, we assume that there is no inter-symbol interference. First, we consider the timing channel proposed in \cite{eck09,sri12}, where the information is encoded in the release timing of a single information particle. Let $T_{x}$ be the release timing of the information particle, and $T_{y}$ be the time of arrival at the receiver. Then we have
\begin{align}
\label{eq:timingChA}
	\tag{A}
	T_{y} = T_{x} + T_n,
\end{align}
where $T_n$ is the random propagation delay of the information particle. $T_n$ is parametrized by the distance between the transmitter and the receiver and the diffusion coefficient of the information particle. 

One of the main challenges of this propagation scheme is the need for synchronization between the transmitter and the receiver. To overcome this challenge, time between release modulation (TBRM) could be used, where information is encoded in the time duration between two consecutive release of molecules. 
Two cases are possible: either the two released information particles are the same, or the two released information particles are different. 

First, we consider the case where both information particles are the same. Let $T_{x_1}$ be the release timing of first information particle and $T_{x_2}$ be the release timing for the second information particle with $T_{x_2}>T_{x_1}$. We assume the information is encoded in $L_x = T_{x_2}-T_{x_1}$. Then using (\ref{eq:timingChA}), the channel model for this modulation scheme is given by:
\begin{align}
\label{eq:timingChB}
	|T_{y_2} - T_{y_1}|&= |T_{x_2}-T_{x_1} + T_{n_2} - T_{n_1}|,  \nonumber \\
	L_y &= | L_x+L_n|, \tag{B}
\end{align}
where $L_n = T_{n_2} - T_{n_1}$ is the random noise and $T_{n_2}$ and $T_{n_1}$ are independent and identically distributed noise terms in (\ref{eq:timingChA}).

Another modulation scheme is when two different types of information particles are used. Let $T_{x_a}$ be the release timing of type-$a$ information particle and $T_{x_b}$ be the release timing for the type-$b$ information particle. We assume the information is encoded in $D_x = T_{x_b}-T_{x_a}$. Unlike (\ref{eq:timingChB}) where $L_x$ is always positive, in this case $D_x$ can be positive or negative depending on the order that type-$a$ and type-$b$ information particles are released. Using (\ref{eq:timingChA}), the channel model for this scheme is given by:
\begin{align}
\label{eq:timingChC}
	T_{y_b} - T_{y_a} &= T_{x_b}-T_{x_a} + T_{n_b} - T_{n_a},  \nonumber \\
	Z_y &= Z_x+ Z_n, \tag{C}
\end{align}
where $Z_n = T_{n_b} - T_{n_a}$ is the random noise and $T_{n_b}$ and $T_{n_a}$ are independent noise terms in (\ref{eq:timingChA}).

\section{Timing Channel Noise Models}
\label{sec:noiseModel}
In this section, we will find the probability density function of the noise terms $T_n$, $L_n$, and $Z_n$ and discuss some of the properties of these random variables.

\subsection{Channel A}

First, we consider the channel in (\ref{eq:timingChA}) and the random propagation noise term $T_n$. If we assume that the receiver is absorbing the information particles, which is the case for many practical applications, $T_n$ is distributed according to the first hitting time distribution. In previous works, it was shown that the first hitting time for the flow induced diffusion in 1-dimensional (1D) space follows Inverse Gaussian distribution \cite{sri12}. In this work, we consider the diffusion channel with no flows. In this case, $T_n$ is a L\'evy distributed random variable. The probability density function (PDF) of a non-negative L\'evy-distributed random variable $X$ is given by
\begin{align}
	f(x;\mu,c) = \sqrt{\frac{c}{2 \pi (x-\mu)^3}}\exp \left( -\frac{c}{2(x-\mu)} \right),
\end{align} 
where $\mu$ and $c$ are parameters of the L\'evy distribution. The characteristic function for a L\'evy distributed random variable is given by
\begin{align}
	\varphi (t;\mu,c) = \exp \left( j \mu t -\sqrt{-2 j c t } \right),
\end{align}
where $j = \sqrt{-1}$ is the imaginary number. We use the notation $\sim\text{L\'evy}(\mu,c)$ to represent a L\'evy distributed random variable with parameters $\mu$ and $c$. Using this notation the additive noise is given by $T_n \sim \text{L\'evy}(0,\frac{d^2}{2D})$, where $d$ is the distance between the transmitter and the receiver and $D$ is the diffusion coefficient. Therefore, we have
\begin{align}
	f_{T_n}(t_n) = \frac{d}{\sqrt{4 \pi D (t_n)^3}}\exp \left( -\frac{d^2}{4Dt_n} \right),
\end{align} 
Similarly, the conditional PDF $P(T_{y}|T_{x}) \sim \text{L\'evy}(T_{x},\frac{d^2}{2D})$. The L\'evy distributed noise holds for 1D diffusion and also for 3D diffusion with a spherical absorbing receiver with an scaling parameter \cite{yilmaz20143dChannelCF,Farsard_arXiv14}.

\subsection{Channel B}

To find the noise distribution for the channel in (\ref{eq:timingChB}), we consider a class of probability distributions known as {\em stable distributions} \cite{zol86-book,nol15}. The L\'evy distribution is a part of stable distributions.
\begin{defn}
	A random variable $X$ has a stable distribution if for two independent copies $X_1$ and $X_2$, and positive constants $a$, $b$, $c$, and $d \in \mathbb{R}$ the following holds
\begin{align*}
	aX_1 + bX_2 \overset{d}{=} cX+d,
\end{align*}
where $\overset{d}{=}$ is equality in distribution.
\end{defn} 

Generally, stable distributions are defined by their characteristic function
\begin{align}
	\varphi (t;\mu,c,\alpha,\beta) = \exp \left[ j \mu t -| c t |^\alpha (1-j\beta \sgn(t) \Phi ) \right],
\end{align}
where $\sgn(.)$ is the sign function (i.e. sign of $t$), $-\infty<\mu<\infty$, $c\geq0$, $0<\alpha \leq 2$, $-1\leq \beta \leq 1$, and 
\begin{align}
\label{eq:phi}
\Phi =
\begin{cases}
	\tan (\pi \alpha / 2) &  \mbox{if } \alpha \neq 1 \\
	-\frac{2}{\pi} \log( |t|) & \mbox{if } \alpha = 1
\end{cases}.
\end{align}
Gaussian distribution, belongs to this family of distributions with $\alpha =2$, and L\'evy distribution with $\alpha = 1/2$ and $\beta=1$. We use the notation $\sim \mathcal{S}(\mu,c,\alpha,\beta)$ to represent a stable distribution with parameters $\mu$, $c$, $\alpha$, and $\beta$. The following are some of the important properties of stable distributions.

\begin{proper}
	If a random variable $X \sim \mathcal{S}(\mu,c,\alpha,\beta)$, and random variable 
	\begin{align*}
		Y = \frac{X-\mu}{c},
	\end{align*}
	then $f(x)dx =f(y)dy$, and $Y$ is the standard form of $X$.	
\end{proper}

\begin{proper}
	Stable random variables with $\beta=0$ have symmetric PDFs.	
\end{proper}

With these definitions we now model the noise term $L_n$ in (\ref{eq:timingChB}). 
\begin{theorem}
The characteristic function for the noise term $L_n$ is given by
\begin{align*}
\varphi \left( t;\frac{\sqrt{2}d}{\sqrt{D}} \right) =  \exp \left[ -\frac{\sqrt{2}d}{\sqrt{D}}\sqrt{| t |} \right],
\end{align*}
where $d$ is the distance between the transmitter and the receiver and $D$ is the diffusion coefficient of the information particle. Therefore, $L_n \sim \mathcal{S}(0,\frac{2d^2}{D},\frac{1}{2},0)$.
\end{theorem}
\begin{proof}
Since $L_n = T_{n_2}+( - T_{n_1})$ with $T_{n_2}, T_{n_1} \sim \mathcal{S}(0,c,\frac{1}{2},1)$, where $c = \frac{d^2}{2D}$. Since $T_{n_1}$ and $T_{n_2}$ are independent, the characteristic function for $L_n$ is given by
\begin{align}
\varphi_{L_n} (t) &=  \varphi_{T_{n_2}} (t) \varphi_{T_{n_1}} (-t) \\
						&= \exp \left[ -\sqrt{| c t |} (1- j \sgn(t)) \right]  \times \nonumber \\ 
						& ~~~~~~~~~~~~~~~~~~~~\exp \left[ -\sqrt{| c t |} (1+ j \sgn(t)) \right] \\
						&= \exp \left[ -\sqrt{|4c t |}  \right]
\end{align}
\end{proof}

Only the PDFs of three classes of stable distributions are known to have closed-form expressions in terms of elementary functions: the Gaussian distribution with $\alpha=2$ (the value of $\beta$ does not matter in this case and can be assumed to equal zero), the L\'evy distribution with $\alpha =0.5$ and $\beta=1$, and Cauchy distribution with $\alpha=1$ and $\beta=0$. To find an expression for the PDF of the noise term $L_n$ in (\ref{eq:timingChB}), we use Property 1, and define the PDF for the standardized distribution with $\frac{2d^2}{D}=1$. Using Property 1, the standard PDF could be used to calculate probabilities involving non-standard random variables just like the way the standard Gaussian PDF could be used to calculate probabilities involving non-standard Gaussian random variables.

The PDF of the standardized stable distribution can be represented by the integral \cite{cul61}
\begin{align}
f(x;\alpha,\beta) = \frac{1}{\pi} \int_0^\infty e^{-t^\alpha} \cos[xt+\beta t^\alpha \Phi],
\end{align}
where $\Phi$ is given in (\ref{eq:phi}). This integral reduces to \cite{cul61}
\begin{align}
\label{eq:alphaHalf}
f(x;1/2,\beta) = \Re \left\{ \frac{z}{\pi x} [ \sqrt{\pi} e^{-z^2} - 2 j F(z)] \right\},
\end{align}
where
\begin{align}
F(z) = e^{-z^2} \int_0^z e^{t^2} dt
\end{align}
is the Dawson's Integral \cite{nist10}, and 
\begin{align}
\label{eq:z}
z = \frac{1+\beta - j(1-\beta)}{2\sqrt{2x}}.
\end{align}
It is possible to rewrite (\ref{eq:alphaHalf}) in terms of the complex error function, also known as Faddeeva function or the Kramp function \cite{nist10}
\begin{align}
w(z) = e^{-z^2} \left( 1+ \frac{2 j}{\sqrt{\pi}} \int_0^z e^{t^2} dt \right) = e^{-z^2} \erfc(-jz),
\end{align} 
where $\erfc(.)$ is the complementary error function. Using the relation \cite{nist10}
\begin{align}
\label{eq:dawsonFadev}
F(z) =0.5 j \sqrt{\pi}( e^{-z^2} -w(z)),
\end{align}
and the property $w(-z) = 2 e^{-z^2} - w(z)$ , we can rewrite (\ref{eq:alphaHalf}) as
\begin{align}
\label{eq:alphaHalf2}
f(x;1/2,\beta) = \Re \left\{ \frac{z}{\sqrt{\pi} x} w(-z) \right\}.
\end{align}
One of the benefits of writing the PDF in terms of the complex error function is that there are a large body of work that considered calculating it numerically. Moreover, if $z = a+j b$, for $b>0$ the complex error function can be represented by its real and imaginary parts as
\begin{align}
w(a+ j b) = K(a,b)+j L(a,b),~~~b>0,
\end{align}
where 
\begin{align}
\label{eq:ReVoigt}
K(a,b) = \frac{1}{\sqrt{\pi}} \int_0^\infty \exp(-t^2/4)\exp(-bt)\cos(at),~~~b>0 
\end{align}
and
\begin{align}
\label{eq:ImVoigt}
L(a,b) = \frac{1}{\sqrt{\pi}} \int_0^\infty \exp(-t^2/4)\exp(-bt)\sin(at),~~~b>0 
\end{align}
are the real and imaginary Voigt functions which are used widely in many fields of physics, astronomy, and chemistry and can be computed numerically. 

Using Property 2, the probability density function of $L_n$ is symmetric. Therefore, the probability density function for $L_n\geq0$ is sufficient for characterizing the PDF. Since $\beta = 0$, when $L_n>0$ we can write $z = p_{l_n} - j p_{l_n}$ where $p_{l_n}= 1/\sqrt{8l_n}$. Using (\ref{eq:alphaHalf2})-(\ref{eq:ImVoigt}) the standardized noise term $L_n$ when $L_n\geq0$ has the PDF 
\begin{align}
\label{eq:fLn}
f(l_n) = 
\begin{cases}
\frac{1}{\sqrt{8 \pi l_n^3}}\left[ K(-p_{l_n},p_{l_n})+L(-p_{l_n},p_{l_n})\right] & l_n>0 \\
\frac{2}{\pi} & l_n=0
\end{cases},
\end{align}
where the second term follows from \cite{zol86-book}. The PDF for $L_n<0$ is then given by $f(-l_n)$ due to symmetry. 

\subsection{Channel C}
The channel noise $Z_n$ given in (\ref{eq:timingChC}) can be different from $L_n$ since two different types of information particles can be used with different diffusion coefficients. Let $D_a$ be the diffusion coefficient of information particle $a$ and $D_b$ be the diffusion coefficient for the information particle $b$. Also, without loss of generality assume particle $a$ is released first followed by particle $b$. We now model the noise term $Z_n$ in (\ref{eq:timingChC}).
\begin{theorem}
The characteristic function for the noise term $Z_n$ is given by 
\small
\begin{align*}
&\varphi \left( t;\frac{d^2(\sqrt{D_a}+\sqrt{D_b})^2}{2D_aD_b}, \frac{\sqrt{D_a}-\sqrt{D_b}}{\sqrt{D_a}+\sqrt{D_b}}\right) =  \\
&~~\exp \left[ -\frac{d(\sqrt{D_a}+\sqrt{D_b})}{\sqrt{2D_aD_b}}\sqrt{| t |} \left( 1- j \frac{\sqrt{D_a}-\sqrt{D_b}}{\sqrt{D_a}+\sqrt{D_b}}\sgn(t)    \right) \right],
\end{align*}
\normalsize  where $d$ is the distance between the transmitter and the receiver and $D_a$ and $D_b$ are the diffusion coefficient of the information particles. Therefore, 
\begin{align*}
Z_n \sim \mathcal{S} \left( 0,\frac{d^2(\sqrt{D_a}+\sqrt{D_b})^2}{2D_aD_b},\frac{1}{2},\frac{\sqrt{D_a}-\sqrt{D_b}}{\sqrt{D_a}+\sqrt{D_b}} \right).
\end{align*}
\end{theorem}
\begin{proof}
Since $Z_n = T_{n_b}+( - T_{n_a})$ with $T_{n_a}, T_{n_b} \sim \mathcal{S}(0,c_i,\frac{1}{2},1)$, where $c_i = \frac{d^2}{2D_i}$ for $i \in\{ a,b \}$. Since $T_{n_a}$ and $T_{n_b}$ are independent, the characteristic function for $Z_n$ is given by
\begin{align}
\varphi_{Z_n} (t) &=  \varphi_{T_{n_b}} (t) \varphi_{T_{n_a}} (-t) \\
						&= \exp \left[ -\sqrt{c_b}\sqrt{| t |} (1- j \sgn(t)) \right]  \times \nonumber \\ 
						& ~~~~~~~~~~~~~~~\exp \left[ -\sqrt{c_a}\sqrt{| t |} (1+ j \sgn(t)) \right] \\
						&= \exp \left[ -\sqrt{|t|}(\sqrt{c_b}+\sqrt{c_a}-j\sgn(t)(\sqrt{c_a}-\sqrt{c_b}))  \right] \\
						&= \exp \left[ -(\sqrt{c_b}+\sqrt{c_a})\sqrt{|t|} \right. \nonumber \\ 
						& ~~~~~~~~~~~~~~~\left. \left( 1-j\sgn(t)\frac{\sqrt{c_a}-\sqrt{c_b}}{\sqrt{c_b}+\sqrt{c_a}}\right)  \right] \\
						&= \exp \left[ -\frac{d(\sqrt{D_a}+\sqrt{D_b})}{\sqrt{2D_aD_b}}\sqrt{| t |} \right. \nonumber \\ 
		&~~~~~~~~~~~~~~~~~~~~\left. \left( 1- j \frac{\sqrt{D_a}-\sqrt{D_b}}{\sqrt{D_a}+\sqrt{D_b}}\sgn(t) \right)\right] .
\end{align}
\end{proof} 

When the diffusion coefficients of the two particles are almost equal, $Z_n$ has the same distribution as $L_n$ (i.e. $\beta=0$). When $\sqrt{D_a}\ll\sqrt{D_a}$ or $\sqrt{D_a}\gg\sqrt{D_a}$, $\beta=\pm1$ and hence $Z_n$ is L\'evy distributed. Therefore, channel (C) can be reduced to channel (A), when one information particle has a much higher diffusion coefficient than the other, with the added benefit that no synchronization is required between the transmitter and the receiver.

For the general case, we have $\beta = (\sqrt{D_a}-\sqrt{D_b})/(\sqrt{D_a}+\sqrt{D_b})$. We can write (\ref{eq:z}) as $z=p_{x} - jq_{x}$ when $x>0$, where $p_{x}=(1+\beta)/(\sqrt{8|x|})$ and $q_{x}=(1-\beta)/(\sqrt{8|x|})$. Similarly, we can write (\ref{eq:z}) as $z=-q_{x} - jp_{x}$ when $x<0$. Then using (\ref{eq:alphaHalf2}) and the Voigt functions decomposition of the Faddeeva function (\ref{eq:ReVoigt}) and (\ref{eq:ImVoigt}) the PDF of the standardized distribution is given by

\begin{align}
\label{eq:fZn}
f(z_n; \beta) = 
\begin{cases}
\frac{1}{\sqrt{8 \pi z_n^3}} \bigg[ (1+\beta) K(-p_{z_n},q_{z_n})  &  \\
\quad \quad  \quad ~~+(1-\beta)L(-p_{z_n},q_{z_n})\bigg], & z_n>0\\
\frac{2(1-\beta^2)}{\pi(1+\beta^2)^2}, & z_n=0 \\
\frac{1}{\sqrt{8 \pi |z_n|^3}}\bigg[(1-\beta)K(q_{z_n},p_{z_n}) &  \\
\quad \quad \quad  ~~-(1+\beta) L(q_{z_n},p_{z_n}) \bigg], & z_n<0
\end{cases},
\end{align}
where the second term follows from \cite{zol86-book}.


\section{Numerical Evaluation and Tail Probabilities}
\label{sec:numEval}
As was shown in the previous section, it is possible to write the PDF for the noise terms in channel (\ref{eq:timingChB}) and (\ref{eq:timingChC}) in terms of real and imaginary Voigt functions. These functions can be numerically calculated using efficient algorithms such as \cite{zag11}. Moreover, for the case of general stable distributions with any parameters $\mu$, $c$, $\alpha$, and $\beta$ it is possible to calculate the PDFs and the cumulative distribution functions (CDF)s numerically using the fast Fourier transform or by numerically solving definite integrals \cite{nol97}. In this section, we plot the PDF and CDF of the noise terms of channels (\ref{eq:timingChA}-\ref{eq:timingChC}) and compare the PDF to the Gaussian PDF, which is typically assumed in the literature.
\begin{figure}
	\begin{center}
	\includegraphics[width=3.4in]{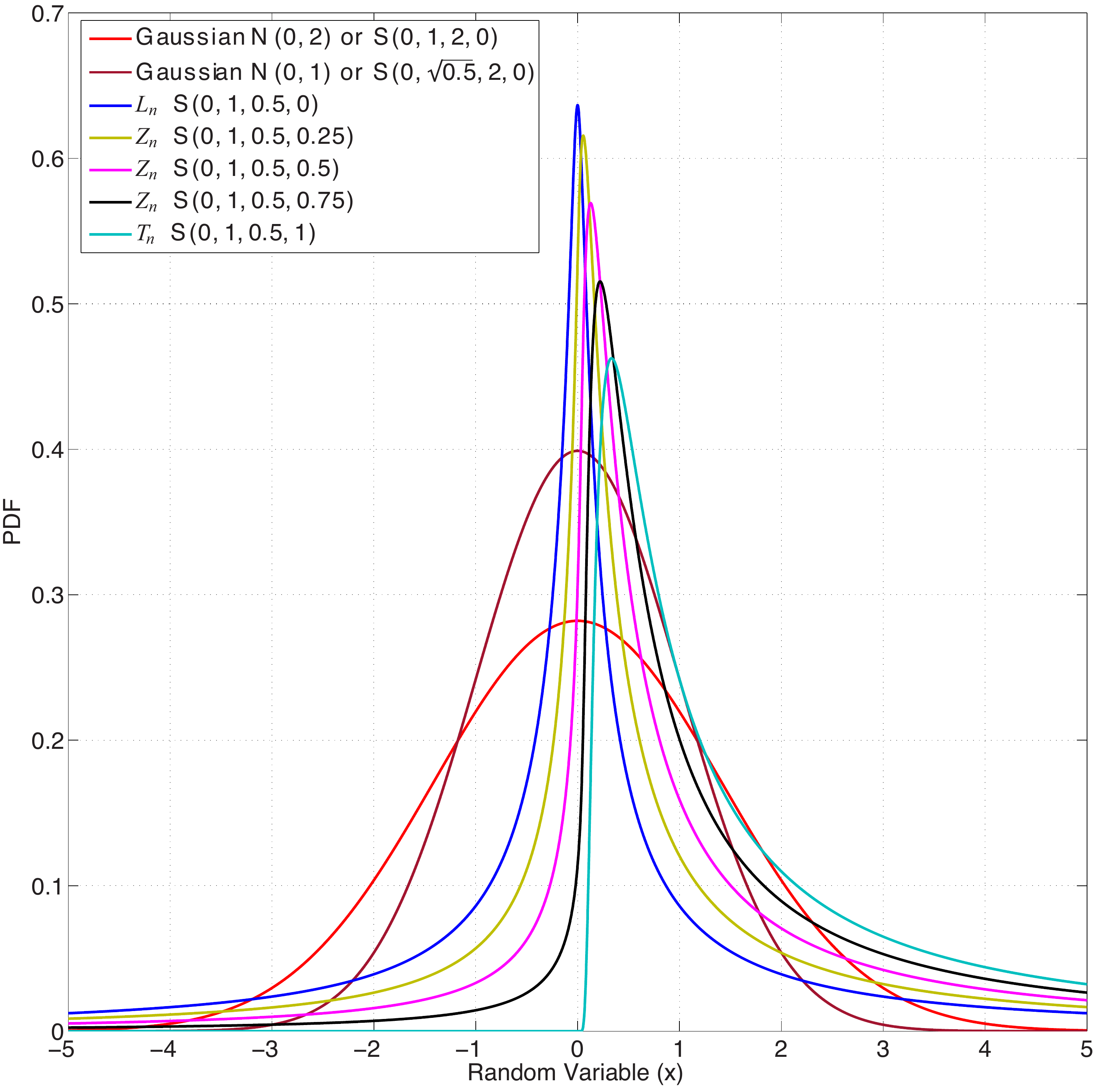}
	\end{center}
	\caption{\label{fig:stablePDF} The probability density function of different standardized stable distributions.}
\end{figure}
\begin{figure}
	\begin{center}
	\includegraphics[width=3.4in]{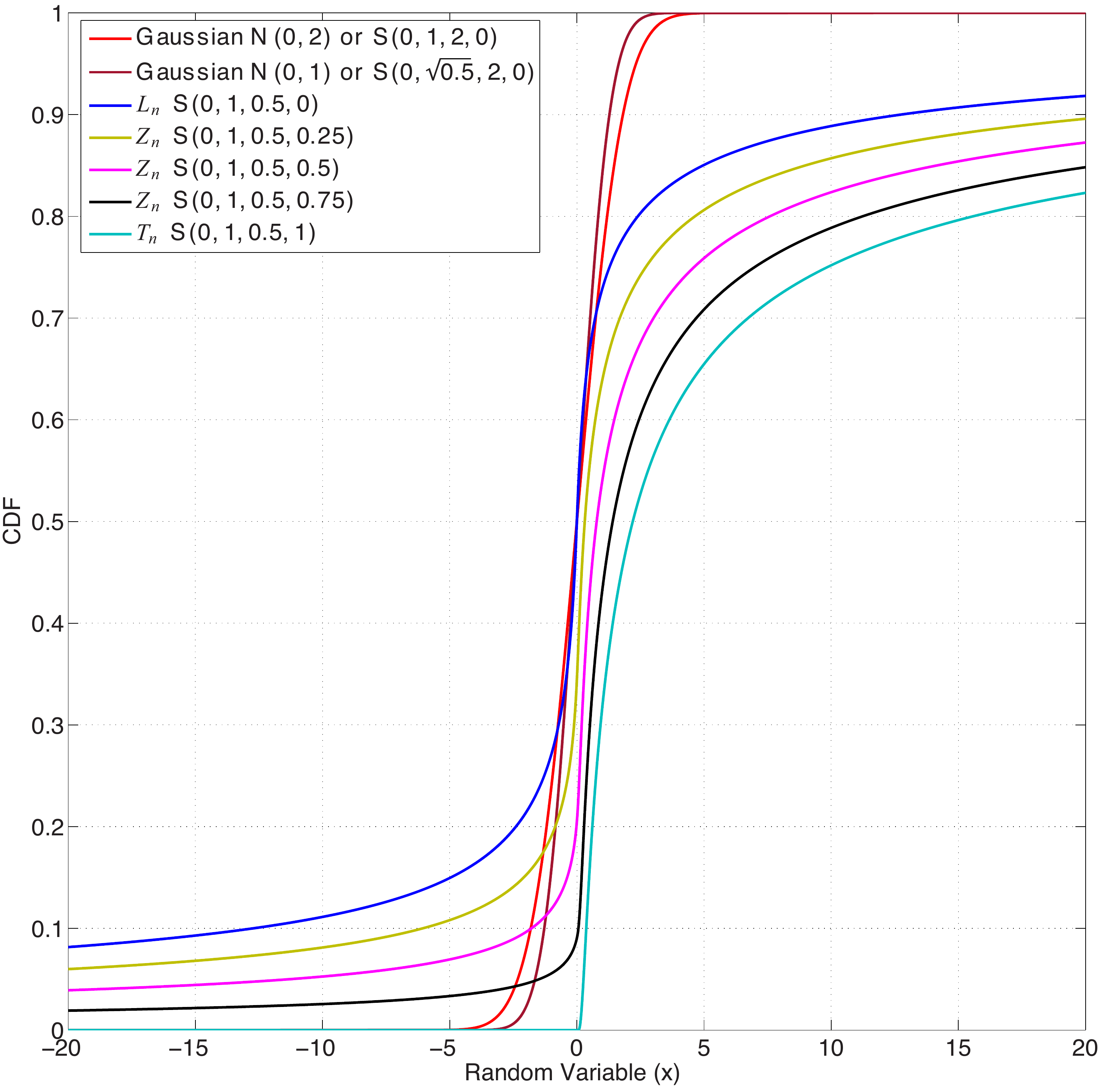}
	\end{center}
	\caption{\label{fig:stableCDF} The cumulative distribution function of different standardized stable distributions.}
\end{figure}

Fig.~\ref{fig:stablePDF} shows the PDF of standardized stable distribution noise terms in channels (\ref{eq:timingChA})-(\ref{eq:timingChC}), as well as the standard Gaussian distribution. As can be seen from the plots, the PDF of the noise terms in all three channels are very different from Gaussian noise. The peaks in the PDF tend to be narrower, while the tails tend to be longer. Moreover, as can be seen the larger the parameter $\beta$ the more asymmetric the PDF. For $\beta=1$ the PDF is the standard L\'evy distribution which is non-zero only for positive values. The cumulative distribution function (CDF) of the standardized stable distributions are shown in Fig.~\ref{fig:stableCDF}. Again it can be seen that each distribution is quite different and that non-Gaussian stable distributions exhibit long tails. 

To compare the tails of each distribution, we use the asymptotic approximation presented in \cite{fof99}. In particular, if $X$ is a standardized stable random variable with parameters $0<\alpha<2$ and $\beta$, then as $x\rightarrow\infty$,
\begin{align}
	P(X>x;\alpha,\beta) \approx \frac{1+\beta}{\pi x^\alpha} \Gamma(\alpha)\sin\bigg(\frac{\alpha \pi}{2}\bigg). 
\end{align} 
For the noise terms in our channel model, $\alpha=1/2$ and hence
\begin{align}
\label{eq:tailApprxS}
	P(X>x;0.5,\beta) \approx \frac{1+\beta}{\sqrt{2 \pi x}},
\end{align} 
as $x\rightarrow\infty$. For the standard normal distribution, the tail probability is approximately,
\begin{align}
\label{eq:tailApprxG}
	P(X>x;2,0) \approx \frac{\exp(-x^2/2)}{x\sqrt{2\pi}},
\end{align} 
as  $x\rightarrow\infty$. This proves the longer tails of non-Gaussian stable distributions.

To measure the accuracy of this approximation, in Fig.~\ref{fig:stableTail} we plot the tail probability $P(X>x)$ and the approximate tail probabilities from (\ref{eq:tailApprxS}) and (\ref{eq:tailApprxG}). In the plot the circles indicate the approximate values. It can be seen that the asymptotic approximation of the tail probabilities quickly converge to the actual probability.  
\begin{figure}
	\begin{center}
	\includegraphics[width=3.4in]{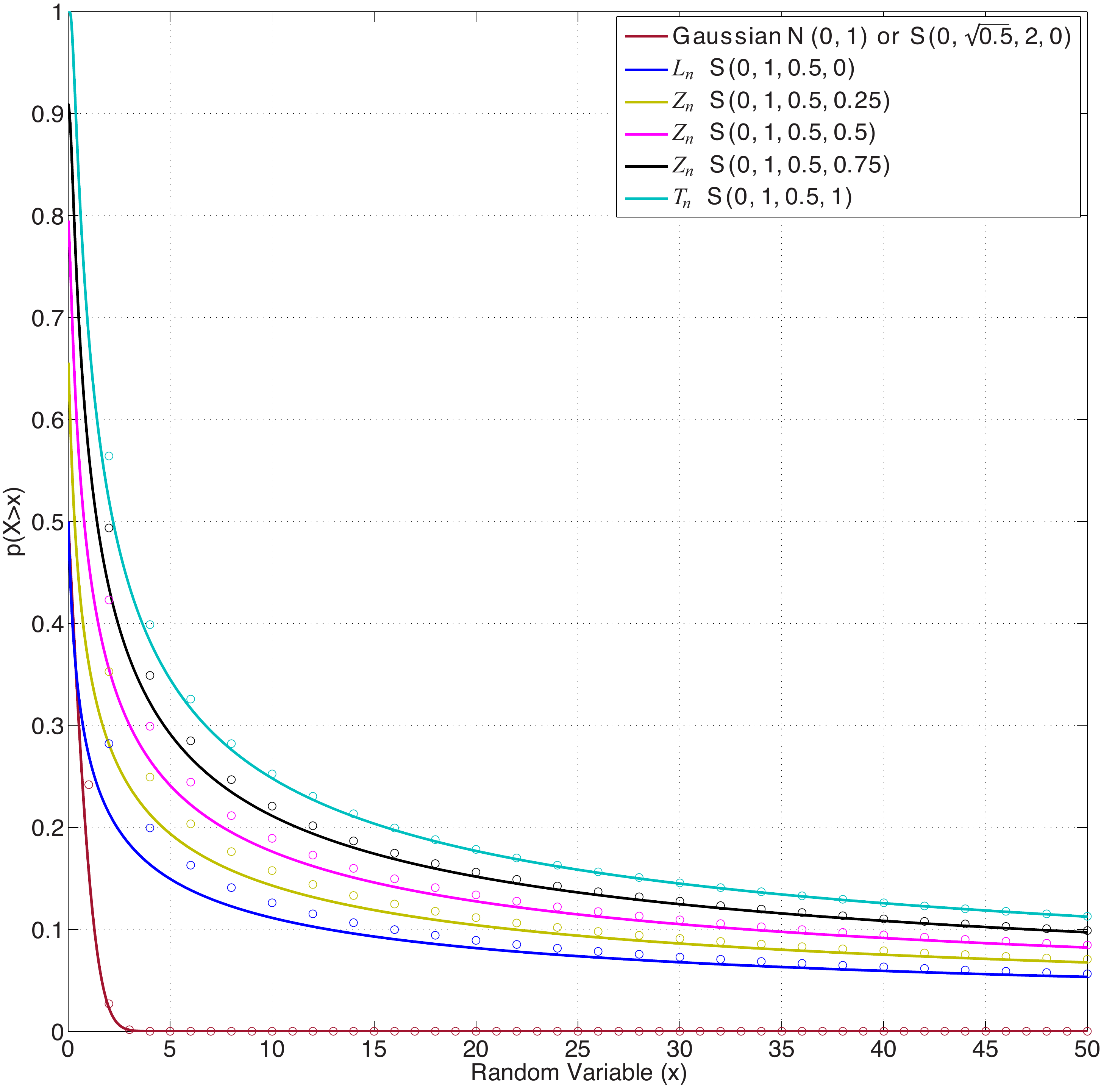}
	\end{center}
	\caption{\label{fig:stableTail} The tail of different stable random variables.}
\end{figure}

\section{Conclusions}
\label{sec:conc}
In this paper, we considered diffusion-based molecular communication timing channels. In particular, we considered three different class of molecular where the information is encoded in the: time of release of information particles, the time between release of two similar information particles, and the time between two different information particles. The channel models for all three different classes were presented as an additive noise channels. It was shown that the noise in all three classes are stable distributed random variables. As a consequence the noise have longer tails, and the effects of inter-symbol interference (ISI) can be more severe than additive Gaussian noise channels. Another interesting observation is that channel (C) can be reduced to channel (A), when one information particle has a much higher diffusion coefficient than the other, with the added benefit that no synchronization is required between the transmitter and the receiver.
As part of future work we will consider finding and comparing the probability of bit error for all three classes of timing channel presented.

\section*{Acknowledgment}
The authors would like to thank Professor John P. Nolan at American University for providing valuable correspondence on stable distributions.

\bibliographystyle{IEEEtran}
\bibliography{IEEEabrv,MolCom}

\end{document}